\documentclass[12pt]{l4dc2021} 
\usepackage[utf8]{inputenc}
\usepackage{color}

\newtheorem{problem}{Problem}

\DeclareMathOperator*{\argmin}{argmin}

\DeclareMathOperator*{\cost}{cost}

\title{Regret-optimal measurement-feedback control}
\author{%
 \Name{Gautam Goel} \Email{ggoel@caltech.edu}\\
 \addr Caltech
 \AND
 \Name{Babak Hassibi} \Email{hassibi@caltech.edu}\\
 \addr Caltech%
}
\begin{document}

\maketitle

\begin{abstract}
    We consider measurement-feedback control in linear dynamical systems from the perspective of regret minimization. Unlike most prior work in this area, we focus on the problem of designing an online controller which competes with the optimal dynamic sequence of control actions selected in hindsight, instead of the best controller in some specific class of controllers. This formulation of regret is attractive when the environment changes over time and no single controller achieves good performance over the entire time horizon. We show that in the measurement-feedback setting, unlike in the full-information setting, there is no single offline controller which outperforms every other offline controller on every disturbance, and propose a new $H_2$-optimal offline controller as a benchmark for the online controller to compete against. We show that the corresponding regret-optimal online controller can be found via a novel reduction to the classical Nehari problem from robust control and present a tight data-dependent bound on its regret. 
\end{abstract}

\section{Introduction} 
The central question in control theory is how to regulate the behavior of an evolving system with state $x$ that is perturbed by a  disturbance $w$ by dynamically adjusting a control action $u$; in the measurement-feedback setting, the information available to the controller is restricted to observations $y$ which are corrupted by noise $v$.
Traditionally, this question has been studied in two distinct settings: in $H_2$ control, the disturbance $w$ and noise $v$ are assumed to be generated by stochastic processes and the controller is designed so as to minimize the \textit{expected} control cost, whereas in $H_{\infty}$ control $w$ and $v$ are assumed to be generated adversarially and the controller is designed to minimize the \textit{worst-case} control cost. Both $H_2$ and $H_{\infty}$ controllers suffer from an obvious drawback: they are designed with respect to a specific class of disturbances, and if the true disturbances fall outside of this class, may exhibit poor performance. Indeed, the loss in performance can be arbitrarily large if the disturbances are carefully chosen \cite{doyle1978guaranteed}.

This observation naturally motivates the design of \textit{adaptive} controllers, which dynamically adjust their control strategy as they sequentially observe the disturbances instead of blindly following a prescribed strategy. This problem has attracted much recent attention in machine learning (e.g. \cite{goel2019online, goel2020regret, hazan2020nonstochastic, cohen2019learning, foster2020logarithmic, abbasi2011regret, agarwal2019online, dean2018regret}), mostly from the perspective of regret minimization. In this framework, the online controller is chosen so as to minimize the difference between its cost and the best cost achievable in hindsight by a controller from some fixed class of controllers. The resulting controllers are adaptive in the sense that they seek to minimize cost irrespective of how the disturbances are generated.

In this paper, we take a somewhat different approach to the design of adaptive controllers. Instead of designing a controller to minimize regret against the best controller selected in hindsight from some specific class, we instead focus on designing a controller which minimizes regret against the \textit{optimal dynamic sequence of control actions selected in hindsight}. We believe that this formulation of regret minimization in control is more attractive than the standard formulation, where the controller learns the best fixed controller in some specific class, for two fundamental reasons. Firstly, it is more general: instead of restricting our attention to some specific class of controllers (e.g. state feedback, LTI controllers, etc), we instead try to compete with the globally optimal dynamic sequence of control actions, without assuming any specific structure. Secondly, and perhaps more importantly, the controllers we obtain are more likely to perform well in dynamic environments, where the disturbance-generating process varies over time. Consider, for example, a scenario in which the disturbances alternate between being generated by a stochastic process and being generated adversarially. When the disturbances are stochastic, an optimistic controller (such as the $H_2$ controller) will perform well; conversely, when the disturbances are adversarial, a more conservative, pessimistic controller (such as an $H_{\infty}$ controller) will perform well. No single controller will perform well over the entire time horizon; hence any online algorithm which tries to learn the best static controller will incur high cumulative cost. A controller which minimizes regret against the optimal dynamic sequence, however, is not constrained to converge to any static controller, and hence can potentially outperform standard regret-minimizing control algorithms  when the environment is dynamic.

Several recent papers \cite{goel2017thinking, goel2019online, goel2020power, goel2020regret} also consider the problem of designing controllers which compete with the optimal offline dynamic sequence of control actions. All of these papers focus on the \textit{full-information} setting, where the controller observes the true state $x$ and disturbance $w$. This paper is the first to study this problem in the more challenging \textit{measurement-feedback} setting, where the controller only has access to a noisy measurement $y$ of the state $x$. This setting presents several unique challenges 
which do not arise in the full-information setting. The key distinction is that \textit{in the measurement-feedback setting, the information sequence observed by a controller depends on the previous control actions selected by that controller}. In essence, the controller is caught in a feedback loop: its control actions depend on the observations it makes, but those observations depend on its previous control actions. These feedback loops make it challenging to analyze control through the lens of regret, since the premise of regret is to compare online policies (which receive information sequentially) to counterfactual offline policies (which receive the \textit{same} information, but all at once, at the start of the game).

\subsection{Contributions of this paper}
We make two main contributions in this paper. First, we consider measurement-feedback control in the offline (noncausal) setting, where the offline controller can compute the measurements $y$ that would counterfactually arise if the offline controller were to select some control $u$. We show that there does not exist a single ``globally optimal" offline measurement-feedback controller which always achieves lower cost than every other offline measurement-feedback controller. This stands in stark contrast to the full-information setting, where a single offline controller dominates every other \cite{goel2020regret}. We derive a new offline controller $K_{nc}$ which optimal in the $H_2$-sense. Second, we consider the problem of designing an online (causal) controller $K_c$ which minimizes regret against the offline controller $K_{nc}$. We show that $K_c$ can be found using a novel reduction to the Nehari problem, which attracted much attention in the robust control community starting in the 1970's. We completely characterize $K_c$ in terms of the solutions to the Nehari problem and present a tight data-dependent bound on its regret.

\subsection{Related work}
There has been a surge of interest in regret minimization in control in the past few years, to the point that we are able to survey only a tiny fraction of the papers in this area. One of the first works in this area was \cite{abbasi2011regret}, which focused on regret minimization when the noise is stochastic. A more general setting where the noise is stochastic but the costs are adversarial was considered in \cite{cohen2019learning}. A series of more recent papers (e.g. \cite{hazan2020nonstochastic, agarwal2019online, dean2018regret, foster2020logarithmic}) consider the setting where the noise is adversarial. All of these works consider a setting where the online learner is trying to minimize static regret against a fixed benchmark controller, often taken to be a state feedback or LTI controller.

A key distinction between this paper and these papers is that we focus on designing an online controller which competes against an optimal offline dynamic sequence of control actions. This problem was also studied in \cite{goel2019online} (albeit through the lens of \textit{competitive ratio} rather than regret) where it was shown that the Online Balanced Descent algorithm introduced in \cite{chen2018smoothed} could be used to give some performance guarantees in the LQR setting; this result was improved in \cite{goel2019beyond}. We note that the reduction in those works relied crucially on  very strong assumptions about the structure of the dynamics, such as invertiblility of the control matrix. In this paper, we are able to remove all such assumptions and prove results about arbitrary LQR control systems. Our results in measurement-feedback control parallel recent results in the much simpler full-information setting obtained in \cite{goel2020regret}.

\section{Preliminaries}
We consider a linear dynamical system governed by the following evolution equation: 
\begin{equation} \label{ss-dynamics}
x_{t+1} = A_tx_t + B_{u, t} u_t + B_{w, t} w_t.
\end{equation}
Here $x_t \in \mathbb{R}^n$ is a state variable we are interested in regulating, $u_t \in \mathbb{R}^m$ is a control variable which we can dynamically adjust to influence the evolution of the system, and $w_t \in \mathbb{R}^n$ is unknown environmental noise.  We formulate the problem of regulating the system over a finite time horizon $t = 0 \ldots T-1$ as an optimization problem, where the goal is to select the control actions so as to minimize the LQR cost 
\begin{equation} \label{ss-lqr-cost}
\cost(w, u) =  x_T^{\top}Q_T x_t + \sum_{t=0}^{T-1} x_t^{\top}Q_t x_t + u_t^{\top}R_t u_t,
\end{equation}
where $Q_t, R_t \succ 0$ for $t = 0, \ldots T - 1$ and $Q_T \succ 0$ is a terminal cost. The sequence of matrices $\{A_t, B_{u, t}, B_{w, t}, Q_t, R_t\}_{t=0}^{T-1}$ is assumed to be known. We assume without any loss of generality that $x_0 = 0$ and $u_t$ is scaled such that $R_t = I$; we emphasize that this imposes no real restriction, since for all $R_t \succ 0$ we can always rescale $u_t$ so that $R_t = I$.

We distinguish between two types of control problems, depending on what information is available to the controllers. In the \textit{full information} setting, we assume the controller has access to the actual state $x$ and disturbance $w$. In the more challenging \textit{measurement-feedback} setting studied in this paper, we assume that the controller only has access to noisy measurements $y$ of the state $x$:
\begin{align} \label{observation-model}
y_t = C_t x_t + v_t,
\end{align}
where $C_t \in \mathbb{R}^{p \times n}$ and $v_t \in \mathbb{R}^p$. We emphasize that this observation model can represent a significant restriction on the information available to the controller; consider, for example, a scenario where $p \ll n$ and the noise $v_t$ is selected adversarially, so that the controller only has access to compressed, highly corrupted information about the state.

\subsection{Causal, noncausal, and anticausal operators}
We distinguish between two types of controllers: \textit{causal} (online) controllers, which select the control action $u_t$ using only the information up to time $t$, and \textit{noncausal} (offline) controllers, which may select $u_t$ using all the  information over the full time horizon. We say that a linear operator $\mathcal{S}$ is causal if it is block lower-triangular; if $u = \mathcal{S}v$, then each $u_t$ is a linear function of $v_0, \ldots, v_t$, so $u$ is a causal function of $v$. Similarly, we say that $\mathcal{S}$ is strictly anticausal if it is strictly block upper-triangular. We say that $\mathcal{S}$ is noncausal if it is not causal; in this case each $u_t$ may potentially depend on some or all of $v_{t+1}, \ldots, v_{T-1}$. We define $M_{+}$ and $M_{-}$ to be the causal and strictly anticausal components of a matrix $M$, so that $M_{+} + M_{-} = M$. If $M$ is positive definite, we use the notation $M^{1/2}$ to mean the unique causal matrix $L$ such that $L^{\top}L = M$. 

\subsection{The input-output approach to control}
It is convenient to encode the dynamics in ``operator form", instead of the state-space form (\ref{ss-dynamics}).
Let $s_t = Q^{1/2}x_t$ for $t = 0, \ldots T-1$ and define $$u =  \begin{bmatrix} u_0 \\ \vdots\\  u_{T-1} \end{bmatrix}, \hspace{5mm} s = \begin{bmatrix} s_0 \\ \vdots\\  s_{T-1} \end{bmatrix}, \hspace{5mm} y = \begin{bmatrix} y_0 \\ \vdots\\  y_{T-1} \end{bmatrix}, \hspace{5mm} w = \begin{bmatrix} w_0 \\ \vdots\\  w_{T-1} \end{bmatrix},  \hspace{5mm} v = \begin{bmatrix} v_0 \\ \vdots\\  v_{T-1} \end{bmatrix}.$$  




\noindent With this notation, the LQR cost (\ref{ss-lqr-cost}) takes a very simple form: $$\cost(w; u) = \|s\|_2^2 + \|u\|_2^2. $$Clearly $s = Fu + Gw$ and $y = Ju + Lw + v$ where $F, G, J$ and $L$ are appropriately defined strictly causal operators encoding the dynamics (\ref{ss-dynamics}) and observation model (\ref{observation-model}). In this paper, we focus on control strategies where the control $u$ is a \textit{linear} function of the measurements: $u = Ky$, for some matrix $K$, which we think of as a controller mapping observations to control actions. Solving for $y$ in terms of $w, v$, and $K$, we see that $y = (I - JK)^{-1}(Lw + v)$. Define the Youla parameterization $Q = K(I - JK)^{-1}$. Notice that $Q$ is causal if and only if $K$ is causal, and furthermore, we can easily recover $K$ from $Q$: 
\begin{equation} \label{kfromq}
K = Q(I + QJ)^{-1}.
\end{equation}
Recall that every controller $K$ has an associated transfer operator  $$\mathcal{T}_{K}:   \begin{bmatrix} w \\ v \end{bmatrix}  \rightarrow \begin{bmatrix} s \\ u \end{bmatrix}. $$
We can write $\mathcal{T}_K$ in terms of $F, G, L$, and $Q$ as $$\mathcal{T}_K = \begin{bmatrix}  FQL + G & FQ \\ QL &  Q \end{bmatrix}.$$ 
We can write the LQR cost incurred by the controller $K$ on the instance $(w, v)$ as $$\cost(K, w, v) =  
\begin{bmatrix} w \\ v \end{bmatrix}^{\top}\mathcal{T}_K^{\top} \mathcal{T}_K \begin{bmatrix} w \\ v \end{bmatrix}. 
 $$



\subsection{$H_{\infty}$-optimal control and regret-optimal control}
Our approach to regret-optimal control is strongly influenced by classic techniques from robust control, whose central objective is the design of $H_{\infty}$-optimal controllers:

\begin{problem}[$H_{\infty}$-optimal measurement-feedback control] \label{hinf-optimal-control-problem} Find a causal controller $K$ that minimizes $$\sup_{w, v} \frac{\cost(K, w, v)}{ \|w\|_2^2 + \|v\|_2^2}.  $$
\end{problem}

\noindent This objective has the natural interpretation of minimizing the worst-case cost incurred by the online controller, normalized by the energy in the disturbance $w$ and noise $v$. In this paper, instead of minimizing the worst-case \textit{cost}, our goal is to minimize the worst-case \textit{regret}.  This problem has a natural analog of the $H_{\infty}$ problem:

\begin{problem}[Regret-optimal control problem] \label{regret-optimal-control-problem}
\noindent Given a benchmark controller $K_0$, find a causal controller $K$ that minimizes $$\sup_{w, v} \frac{\cost(K, w, v) - \cost(K_0, w, v)}{{ \|w\|_2^2 + \|v\|_2^2}}.  $$
\end{problem}

As is common in the $H_{\infty}$ literature, we consider the relaxation:

\begin{problem} [Regret-suboptimal control problem] \label{regret-suboptimal-control-problem}
Given a performance level $\gamma > 0$ and a benchmark controller $K_0$, find a causal controller such that $$ \frac{\cost(K, w, v) - \cost(K_0, w, v) }{  \|w_t\|_2^2 + \|v\|_2^2} < \gamma  $$ for all disturbances $w$, or determine whether no such policy exists.
\end{problem}

We emphasize that if we can solve the regret-suboptimal  problem , we can easily recover the solution to the regret-optimal problem via bisection on $\gamma$.

\subsection{The Nehari problem}
A key idea in this paper is to reduce the regret-optimal measurement-feedback control problem to the Nehari problem, which asks how best to approximate an anticausal matrix by a causal matrix:

\begin{problem}[Nehari problem]
Let $W$ be an strictly anticausal matrix. Find a causal matrix $X$ such that $$\|X - W\|_2$$ is minimized. 
\end{problem}
Like the $H_{\infty}$-optimal control problem, the Nehari problem is generally solved by first solving a suboptimal problem at level $\gamma$, and then finding the optimal problem by bisection on $\gamma$: 

\begin{problem}[Suboptimal Nehari problem]
Let $W$ be an strictly anticausal matrix. Given a performance level $\gamma > 0$, find a causal matrix $X$ such that $$\|X - W\|_2 < \gamma, $$ or determine whether no such $X$ exists.
\end{problem}
We emphasize that there exist efficient numerical algorithms to solve the suboptimal Nehari problem. The details of these algorithms are beyond the scope of this paper; we refer the reader to \cite{gohberg1994fast, hassibi1999indefinite} for details. 

\section{Noncausal measurement-feedback controllers} \label{noncausal-sec}
The goal of this paper is to derive an online controller which minimizes regret against the optimal dynamic sequence of control actions selected in hindsight. In the full-information setting, it is clear what this means: we design an online controller which minimizes regret against the sequence of control actions $$u^* = \argmin_u \cost(w, u).$$
It was recently shown in \cite{goel2020regret} that $u^* = K^*w$, where $K^* = -(I + F^{\top}F)^{-1}F^{\top}G$.
We can hence view the optimal dynamic sequence of control actions selected in hindsight as precisely those actions selected by the optimal noncausal controller $K^*$, and design our online controller to compete against this $K^*$. 


In the measurement-feedback setting we study in this paper, it is much less clear which noncausal benchmark controller we should select for the online controller to compete against. This is because \textit{the information sequence observed by any controller $K$ depends on the previous choices of the controller}. This is easy to see from the two relations $$u = Ky, \hspace{5mm} y = Ju + Lw + v.$$
In essence, the controller $K$ is caught in a feedback loop: its control actions depend on the observations it makes ($u = Ky$), but those observations depend on its previous control actions ($y = Ju + Lw + v$). In the full-information setting this issue does not arise: we assume that the controller observes the true disturbance $w$, irrespective of what control actions they previously selected, and can thus define the optimal noncausal controller $K^*$ to be the unique controller which selects the optimal control $u$ in response to $w$. Given that any offline controller will receive a different set of observations than the online controller, which offline controller should we pick for the online controller to compete against? One natural idea, in analogy with the full-information setting, is to select the ``optimal" offline measurement-feedback controller, i.e. one which always incurs less cost than any other offline controller. Our first result is that no such controller exists:

\begin{theorem} \label{nonexistence-thm}
There does not exist a noncausal controller $K$ such that $\cost(K, w, v) <  \cost(K', w, v)$ for all noncausal controllers $K' \neq K$ and all instances $(w, v)$.
\end{theorem}

\noindent In other words, no noncausal controller can guarantee that it achieves the lowest possible cost on every instance $(w, v)$. Before we present the proof of Theorem (\ref{nonexistence-thm}), we state a key lemma which plays a central role in all of the results of this paper:

\begin{lemma} \label{simplifying-lemma}
Define \begin{equation} \label{stuv}
    S = I + FF^{\top}, \hspace{5mm} T = I + F^{\top}F, \hspace{5mm} U =  I + LL^{\top}, \hspace{5mm} V = I + L^{\top}L
\end{equation}
and let
\begin{equation} \label{thetapsi}
\theta = \begin{bmatrix} S^{-1/2} & 0 \\ 0 & T^{-1/2} \end{bmatrix} \begin{bmatrix} I & - F \\ F^{\top} & I \end{bmatrix},  \hspace{5mm}
\psi = \begin{bmatrix} I & L^{\top} \\ -L & I \end{bmatrix}\begin{bmatrix} V^{-1/2} & 0 \\ 0 & U^{-1/2} \end{bmatrix}  .
\end{equation}
Let $K$ be any controller and let $\mathcal{T}_K$ be the transfer operator associated to $K$. The following identity holds
\begin{align} \label{tk-theta-psi}
\theta \mathcal{T}_{K} \psi = \begin{bmatrix} S^{-1/2}GV^{-1/2} &  S^{-1/2}GL^{\top}U^{-1/2} \\ T^{-1/2}F^{\top}GV^{-1/2} &  T^{1/2}QU ^{1/2} + T^{-1/2}F^{\top}GL^{\top}U^{-1/2}
\end{bmatrix}.
\end{align}
\end{lemma}
This lemma is easily verified via direct calculation; its significance is that the matrix $\theta \mathcal{T}_{K} \psi$ depends on $Q$ only in the (2, 2) entry instead of in all four entries, which greatly simplifies our computations. We now return to the proof of Theorem (\ref{nonexistence-thm}):

\vspace{5mm}

\begin{proof}
Suppose by way of contradiction that there was some noncausal controller $K$ such that $\cost(K, w, v) <  \cost(K', w, v)$ for all noncausal controllers $K' \neq K$ and all instances $(w, v)$. This would imply that \begin{align} \label{offline-domination-cond}
 \mathcal{T}_{K'}^{\top}\mathcal{T}_{K'} - \mathcal{T}_K^{\top} \mathcal{T}_K  \succ 0.
\end{align}
Let  $\theta$ and $\psi$ be defined as in (\ref{thetapsi}). Because $\theta$ and $\psi$ are unitary, condition (\ref{offline-domination-cond}) is equivalent to
\begin{align} \label{thetapsi-offline-domination-cond}
(\theta \mathcal{T}_{K'} \psi)^{\top} (\theta \mathcal{T}_{K'} \psi) - (\theta \mathcal{T}_K \psi)^{\top} (\theta \mathcal{T}_K \psi) \succ 0.
\end{align}
In light of Lemma (\ref{simplifying-lemma}), the matrix on the left-hand side of (\ref{thetapsi-offline-domination-cond}) simplifies to
\begin{align*} 
\begin{bmatrix} 0 &  X(Q, Q') \\ X^{\top}(Q, Q') &  Y(Q, Q')
\end{bmatrix}
\end{align*}
where $X(\cdot, \cdot)$ and $Y(\cdot, \cdot)$ are appropriately defined functions of $Q$ and $Q'$. The fact that the (1, 1) block of this matrix is zero ensures that it cannot be positive definite.
\end{proof}

Given that no noncausal controller dominates every other, we are now faced with the question of which noncausal controller we should design our online controller to compete against. In this paper, we choose to benchmark against the noncausal controller which is optimal in the $H_2$ sense. By this, we mean the noncausal controller $K_{nc}$ whose associated transfer operator $$\mathcal{T}_{K_{nc}}:   \begin{bmatrix} w \\ v \end{bmatrix}  \rightarrow \begin{bmatrix} s \\ u \end{bmatrix} $$ is smallest in the Frobenius norm. This noncausal controller is also the one which minimizes the expected LQR cost under the assumption that $w$ and $v$ are both random variables with zero mean and bounded variance. We note that are several other natural choices of benchmark controllers; for example, one could instead choose the more pessimistic $H_{\infty}$-optimal noncausal controller. We leave such comparisons for future work. We prove:

\begin{theorem}
The $H_2$-optimal noncausal controller has the form $K_{nc} = Q(I + QJ)^{-1}$ and associated transfer operator $$\mathcal{T}_{K_{nc}} = \begin{bmatrix}  FQL + G & F \\ QL &  Q \end{bmatrix},$$ where $Q = -T^{-1}F^{\top}GL^{\top}U^{-1}$ and $T$ and $U$ are defined as in Lemma (\ref{simplifying-lemma}).
\end{theorem}

\begin{proof}
Let $\theta$ and $\psi$ be defined as in Lemma (\ref{simplifying-lemma}).
Notice that $\theta$ and $\psi$ are unitary, hence $\|\theta \mathcal{T}_K \psi \|^2_F = \|\mathcal{T}_K \|^2_F$ since the Frobenius norm is unitarily invariant. It hence suffices to minimize $\|\theta \mathcal{T}_K \psi \|^2_F$ over $Q$. Looking at the statement of Lemma (\ref{simplifying-lemma}), we notice that $Q$ appears only in the (2, 2) entry of $\theta \mathcal{T}_K \psi$, so $\|\mathcal{T}_K\|_F^2$ is minimized by choosing $Q$ such that this entry is zero:  $$Q  =  -T^{-1}F^{\top}GL^{\top}U^{-1}.$$ We can easily recover the $H_2$-optimal $K$ from this choice of $Q$ using identity (\ref{kfromq}).






\end{proof}

\section{Derivation of the regret-optimal measurement-feedback controller}
We now turn to the problem of deriving a causal measurement-feedback controller $K_c$ which tracks $K_{nc}$ as closely as possible, where $K_{nc}$ is the $H_2$-optimal noncausal measurement-feedback controller derived in section \ref{noncausal-sec}. We call $K_c$ the regret-optimal measurement-feedback controller, and show that it can be found via a reduction to the classical Nehari problem of robust control.

We give a high-level summary of our proof technique before turning to the derivation of $K_c$. We first solve the regret-suboptimal problem; in other words, for a given performance level $\gamma > 0$, we show how to find a causal controller $K$ such that 
\begin{align} \label{suboptimal-cond}
\frac{\cost(K, w, v) - \cost(K_{nc}, w, v)}{\|w\|_2^2 + \|v\|_2^2} < \gamma^2 
\end{align}
for all disturbances $w$ and all measurement noise $v$, or to determine whether no such $K$ exists. We show this problem is equivalent to a Nehari problem with performance level 1 and an input matrix $W_{-, \gamma}$; once this Nehari problem is solved, we can easily recover the desired $K$. Conversely, if this Nehari problem has no solution the desired $K$ does not exist.  Once the suboptimal problem is solved, the regret-optimal controller is easily found via bisection on $\gamma$, in the same way $H_{\infty}$ controllers can be found once the $H_{\infty}$ suboptimal problem is solved.

We now state our main result:

\begin{theorem}
There exists a regret-suboptimal measurement-feedback controller $K$ at level $\gamma$ exists if and only if there exists a causal matrix $X$ such that $\|X - W_{-, \gamma}\| \leq 1$, where $W_{+, \gamma}$ and $W_{-, \gamma}$ are the causal and strictly anticausal components of the matrix
$$W_{\gamma} = -M_{\gamma}^{-1}T^{-1/2}F^{\top}GL^{\top}U^{-1/2}, $$
 $$M_{\gamma} = (\gamma^{-2}I + \gamma^{-4}T^{-1/2}F^{\top}GV^{-1} G^{\top}FT^{-1/2})^{1/2},$$
and $T, U, V$ are defined as in Lemma (\ref{simplifying-lemma}).
If such an $X$ exists, then $K$ has the following form: $$ K = Q(I + QJ)^{-1}, $$
where $$Q = T^{-1/2}M_{\gamma}^{-1}(X + W_{+, \gamma})U^{-1/2}.$$
The regret-optimal measurement-feedback controller $K_{c}$ is the regret-suboptimal measurement-feedback controller at level $\gamma_{opt},$ where $\gamma_{opt}$ can be found by bisection on $\gamma$.
Furthermore, the regret incurred by $K_c$ against $K_{nc}$ on the instance $(w, v)$ is at most $\gamma^2_{opt}(\|w\|_2^2 + \|v\|_2^2),$ and this bound is tight.
\end{theorem}

\begin{proof}
Condition (\ref{suboptimal-cond}) can be expressed in terms of transfer operators as 
\begin{align*} 
 \begin{bmatrix} w \\ v \end{bmatrix}^{\top}(\mathcal{T}_K^{\top} \mathcal{T}_K - \mathcal{T}_{K_{nc}}^{\top}\mathcal{T}_{K_{nc}})\begin{bmatrix} w \\ v \end{bmatrix} < \gamma^2 \left\|  \begin{bmatrix} w \\ v \end{bmatrix} \right\|_2^2, 
\end{align*}
 or even more cleanly as 
\begin{align} \label{regret-condition}
\mathcal{T}_K^{\top} \mathcal{T}_K - \mathcal{T}_{K_{nc}}^{\top}\mathcal{T}_{K_{nc}} \prec \gamma^2 I.
\end{align}
Let  $\theta$ and $\psi$ be defined as in (\ref{thetapsi}). Because $\theta$ and $\psi$ are unitary, condition (\ref{regret-condition}) is equivalent to
\begin{align} \label{regret-condition-2}
\gamma^2I - (\theta \mathcal{T}_K \psi)^{\top} (\theta \mathcal{T}_K \psi) + (\theta \mathcal{T}_{K_{nc}} \psi)^{\top} (\theta \mathcal{T}_{K_{nc}} \psi) \succ 0
\end{align}

\noindent Applying identity (\ref{tk-theta-psi}), we can expand the matrix on the left-hand side of (\ref{regret-condition-2}) as
\begin{align} \label{regret-matrix}
\begin{bmatrix} \gamma^2 I &  V^{-1/2} G^{\top}FT^{-1/2} (W - T^{1/2}QU ^{1/2}) \\ (W - T^{1/2}QU ^{1/2})^{\top}T^{-1/2}F^{\top}GV^{-1/2} & \gamma^2I - (T^{1/2}QU ^{1/2} -  W)^{\top}(T^{1/2}QU ^{1/2} - W)
\end{bmatrix}
\end{align}
where $T, U, V$ are defined in (\ref{stuv}), we set $W = -T^{-1/2}F^{\top}GL^{\top}U^{-1/2}$, and we used the fact the (2, 2) entry of $\theta \mathcal{T}_{K_{nc}} \psi$ is zero. Clearly $\gamma^2I$ is positive definite, so the matrix (\ref{regret-matrix}) is positive definite if and only if the Schur complement 
\begin{equation*}
\gamma^2I - (T^{1/2}QU ^{1/2} -  W)^{\top}(I + \gamma^{-2}T^{-1/2}F^{\top}GV^{-1} G^{\top}FT^{-1/2}) (T^{1/2}QU ^{1/2} -  W)
\end{equation*}
is positive definite. Dividing by $\gamma^2$ and rearranging, we see that this condition is equivalent to 
\begin{equation} \label{contraction-cond}
(T^{1/2}QU ^{1/2} -  W)^{\top}(\gamma^{-2}I + \gamma^{-4}T^{-1/2}F^{\top}GV^{-1} G^{\top}FT^{-1/2}) (T^{1/2}QU ^{1/2} -  W) \preceq 1.
\end{equation}
We have shown that there exists a causal controller $K$ satisfying (\ref{suboptimal-cond}) if and only if there exists a causal matrix $Q$ satisfying condition (\ref{contraction-cond}).

We now show that the problem of finding such a $Q$ (or determining whether no such $Q$ exists) can be reduced to a Nehari problem with an appropriate change of variables. Define $$M_{\gamma} = (\gamma^{-2}I + \gamma^{-4}T^{-1/2}F^{\top}GV^{-1} G^{\top}FT^{-1/2})^{1/2},$$ and let $ W_{\gamma} = M_{\gamma}W$. Let $W_{\gamma, +}$ and $W_{\gamma, -}$ denote the causal and strictly anticausal parts of $W_{\gamma}$, respectively. Define $$X = M_{\gamma}T^{1/2}QU^{1/2} - W_{+, \gamma}.$$ We emphasize that $X$ is causal if and only if $Q$ is causal, since $W_{+, \gamma}, M_{\gamma}, T^{1/2}$, and $U^{1/2}$ are all causal; furthermore, we can easily recover $Q$ from $X$: 
\begin{align} \label{qfromx}
Q = T^{-1/2}M_{\gamma}^{-1}(X + W_{+, \gamma})U^{-1/2}.
\end{align}
Notice that condition (\ref{contraction-cond}) can be written as $$\| X - W_{-, \gamma}\|_2 \leq 1.$$ 
We recognize the problem of finding a causal $X$ satisfying this condition as an instance of the suboptimal Nehari problem, where the desired performance level is 1 and the input anticausal matrix is $W_{-, \gamma}$. If the desired $X$ exists then $K$ is easily found using the identities (\ref{kfromq}) and (\ref{qfromx}); conversely, if this Nehari problem has no solution then a regret-suboptimal controller at level $\gamma$ does not exist.

Now that we know how to solve the regret-suboptimal problem, we can easily find the regret-optimal controller $K_c$ by iteratively decreasing $\gamma$ until convergence to some $\gamma_{opt}.$ Rearranging condition (\ref{suboptimal-cond}), we immediately obtain the data-dependent regret bound $\gamma^2_{opt}(\|w\|_2^2 + \|v\|_2^2)$; tightness of this bound follows from the minimality of $\gamma_{opt}$.

\end{proof}

\newpage
\bibliography{main}

\end{document}